\documentclass[11pt]{amsart}
\usepackage{tikz-cd}
\usepackage{verbatim}
\usepackage[colorinlistoftodos]{todonotes}
\usepackage{mathtools}
\usepackage[bookmarks,colorlinks,breaklinks]{hyperref}  
\usepackage[top=2cm, bottom=4.5cm, left=3cm, right=3cm]{geometry}
\hypersetup{linkcolor=blue,citecolor=blue,filecolor=blue,urlcolor=blue} 
\usepackage{tikz}\usetikzlibrary{cd,fit,matrix,arrows,decorations.pathmorphing}
\tikzset{commutative diagrams/.cd}
\usepackage{bm}

\numberwithin{equation}{section}

\newtheorem{theorem}{Theorem}[section]

\newtheorem{lemma}[theorem]{Lemma}

\theoremstyle{definition}

\newtheorem{definition-theorem}[theorem]{Definition-Theorem}
\newtheorem{example}[theorem]{Example}

\newtheorem{remark}[theorem]{Remark}

\theoremstyle{remark}

\newtheorem*{remark*}{Remark}

\usepackage{enumitem}
\setenumerate[1]{label=\textup{(\alph*)}}
\setenumerate[2]{label=\textup{(\roman*)}}

\begin{document}
\title{Frequencies of subwords in words of linear subword complexity}

\author{Jason Bell}
\address{University of Waterloo \\
Department of Pure Mathematics \\
Waterloo, Ontario \\
N2L 3G1, Canada}
\email{jpbell@uwaterloo.ca}

\author{Laindon Burnett}
\address{University of Waterloo \\
Department of Pure Mathematics \\
Waterloo, Ontario \\
N2L 3G1, Canada}
\email{lcburnett@uwaterloo.ca}

\author{Chris Schulz}
\address{University of Waterloo \\
Department of Pure Mathematics \\
Waterloo, Ontario \\
N2L 3G1, Canada}
\email{chris.schulz@uwaterloo.ca}


\begin{abstract}
Using a method of Balková--Pelantová, we show that if ${\bf w}$ is a right-infinite word over a finite alphabet, then for each nonnegative integer $N$ there are at most $3(p_{\bf w}(N+1)-p_{\bf w}(N))+1$ distinct upper (and likewise lower and ordinary when they exist) frequencies for length-$(N+1)$ subwords of ${\bf w}$, where $p_{\bf w}(n)$ is the subword complexity function of $n$. In particular, this gives a uniform upper bound when ${\bf w}$ has linearly bounded subword complexity.  We provide examples showing that whenever $f(n)$ is a weakly increasing function tending to infinity, there is a word ${\bf w}$ such that the number of subwords of length $n$ is $O(nf(n))$ and for which the limit supremum of the number of distinct upper frequencies of length-$N$ subwords of ${\bf w}$ as $N\to\infty$ is infinite.\end{abstract}

\maketitle

\section{Introduction}

Let ${\bf w}=u_1u_2u_3\cdots $ be a right-infinite word over a finite alphabet $\Sigma$.  
Given a word $u \in \Sigma^*$, it is natural to study the frequency with which $u$ occurs in ${\bf w}$ as this gives insight into the complexity of the word and the extent to which patterns recur in ${\bf w}$. More precisely, we define the \emph{frequency} of a length-$m$ word $u$ in ${\bf w}$ to be the limit
$$\lim_{n\to\infty} \frac{\#\{p\le n\colon u_{p}\cdots u_{p+m-1}=u\}}{n},$$ if it exists.  When it exists, we let ${\rm freq}_u({\bf w})$ denote this quantity.  If we replace the limit by a limsup or liminf we obtain respectively the \emph{upper frequency} and \emph{lower frequency} of $u$, which we denote
$\overline{{\rm freq}}_u(w)$ and $\underline{{\rm freq}}_u(w)$ respectively, and these quantities always exist.

Frequencies often shed insight into the underlying structure of a word.  For example, an infinite word ${\bf w}$ is periodic if and only if there is some positive integer $p$ such that each length-$p$ subword of ${\bf w}$ occurs with frequency $1/p$.  In general, subword frequencies have some relation with the \emph{subword complexity} function, $p_{\bf w}$ of an infinite (or finite) word, which is a function that takes a natural number $n$ as input and outputs $p_{\bf w}(n)$, the number of subwords of ${\bf w}$ of length $n$.  (This function is zero for $n$ larger than the length of ${\bf w}$ in the case when ${\bf w}$ is finite.) It is clear that if $p_{\bf w}$ grows quickly with $n$ then most length-$n$ subwords cannot occur too frequently.  

In addition, there are several fundamental questions about the complexity of the frequencies as real numbers, which often reveal how well behaved the word is. For example, in the case of automatic words, the frequencies of subwords (when they exist) are always rational \cite{Cob} and the upper and lower frequencies are always rational \cite{Bellupper}.  In the case of morphic words, the frequencies, when they exist, are algebraic numbers (see Allouche and Shallit \cite[Theorem 8.4.5]{AS}).  Numbers with transcendental frequencies can exist. For example, the infinite binary word ${\bf w}=u_1u_2u_3\cdots $ in which $u_n=1$ precisely when there is a positive integer $j$ such that $n=\lfloor j\pi \rfloor$ has the property that $1$ occurs with frequency $1/\pi$; however, words where subwords occur with non-algebraic frequencies are typically considered pathological.

The motivation for this paper arose from a well-known fact about Sturmian words.\footnote{A right-infinite binary word ${\bf w}$ is Sturmian when $p_{\bf w}(n)=n+1$ for all $n\ge 0$ and this is the minimal subword complexity function for words that are not eventually periodic.} It is an immediate consequence of the famous ``three-gap theorem'' of S\'os, Sur\'anyi, and Świerczkowski \cite{Sos, Suranyi, Swier} and the description of Sturmian words via irrational rotations (see, for example, \cite[Chapter 2]{Lot}) that if $d$ is a positive integer then there are at most three distinct values (which depend on $d$) that can occur as frequencies of length-$d$ subwords of a fixed Sturmian word.  This result was extended to uniformly recurrent words by Boshernitzan \cite{Bosh} (see also Balková--Pelantová \cite{Balkova}), who showed that the set of distinct frequencies of length-$N$ words could be bounded in terms of the subword complexity function. We recall that the subword complexity of a right-infinite word is \emph{linearly bounded} if there is a positive constant $C$ such that $p_{\bf w}(n)\le Cn$ for all $n$ sufficiently large.  Boshernitzan's \cite{Bosh} result shows in particular that if ${\bf w}$ is uniformly recurrent and has linearly bounded subword complexity then the number of frequencies of length-$N$ subwords is uniformly bounded.  We adapt a clever combinatorial argument of Balková--Pelantová \cite{Balkova} to extend this result to general words where we instead use upper and lower frequencies, as frequencies need not exist in a general word.

We use the notation $|u|$ to denote the length of a finite word $u$. 
\begin{theorem}
\label{thm:main}
Let ${\bf w}$ be a right-infinite word over a finite alphabet $\Sigma$ that is not eventually periodic.  
Then for each $N\in \mathbb{N}$, the sets 
$$\{\overline{{\rm freq}}_u({\bf w})\colon u\in \Sigma^*, |u|=N+1\}$$
and
$$\{\underline{{\rm freq}}_u({\bf w})\colon u\in \Sigma^*, |u|=N+1\}$$
have size at most $3(p_{\bf w}(N+1)-p_{\bf w}(N))+1$. In particular, if ${\bf w}$ has linearly bounded subword complexity then there is a constant $C$ such that for every $N$, the number of distinct upper and lower frequencies of length-$N$ subwords is bounded above by $C$.
\end{theorem}

A natural follow-up question is to ask what occurs when one looks beyond linearly bounded subword complexity.  In some sense, this uniform boundedness is a generic phenomenon: for example, a famous result of Borel \cite{Bor} shows that almost all real numbers\footnote{``Almost all'' in this context means outside of a set of Lebesgue measure zero.} have base-$k$ expansions that are normal; here normality for a base-$k$ expansion means that each length-$d$ subword has frequency $1/k^d$.  In particular, in a very natural sense we have that almost all right-infinite words have the property that every length-$d$ subword occurs with the same frequency. On the other hand, we give examples that show that a uniform bound on the number of frequencies of subwords of a given length does not exist once one goes even slightly beyond linearly bounded subword complexity.  Specifically, in Example \ref{exam:2}, we show that for every weakly increasing function $f(n)$ that tends to infinity there are words ${\bf w}$ with $p_{\bf w}(n)=O(nf(n))$, where no uniform bound exists on the upper frequencies.  Thus one cannot extend the ``In particular'' clause of Theorem \ref{thm:main} outside the realm of linearly bounded subword complexity. This is actually the most difficult part of the paper, since it is a difficult problem in general to construct examples of words with prescribed subword complexity having pathological properties. Part of the inspiration for our construction comes from an algebraic construction of Vishne \cite{Vishne}. In addition, we show that in general the number of distinct frequencies of length-$d$ subwords of a right-infinite word can grow exponentially in $d$ (see Example \ref{exam:1}).

The outline of this paper is as follows.  In Section \ref{sec:proof} we prove Theorem \ref{thm:main} by adapting an argument of Balková--Pelantová \cite{Balkova} using Rauzy graphs and by employing a result of Cassaigne \cite[Th\'eor\`eme 7.2]{Cass}.  In Section \ref{sec:exam} we give examples of infinite words that illustrate that the uniform bound on the number of frequencies of length-$d$ subwords is, in general, linked to the subword complexity.  

\section{Proof of Theorem \ref{thm:main}}\label{sec:proof}
In this section, we let $\Sigma=\{x_1,\ldots ,x_d\}$ denote a finite alphabet and we let ${\bf w}$ denote a right-infinite word. We note that in the special case that the complexity function $p_{\bf w}(n)$ of ${\bf w}$ is linearly bounded in $n$,
Cassaigne \cite[Th\'eor\`eme 7.2]{Cass} proves there is a positive constant $B$ such that
\begin{equation}
\label{eq:in}
\Delta p_{\bf w}(n) := p_{\bf w}(n+1)-p_{\bf w}(n) \le B\end{equation} for all $n$. 
He also proves that $B$ can be bounded explicitly in terms of $|\Sigma|$ and the rate of growth.

We shall prove Theorem \ref{thm:main}.
\begin{proof}[Proof of Theorem \ref{thm:main}]
Since ${\bf w}$ is not eventually periodic, the Morse--Hedlund theorem gives $\Delta p_{\bf w}(n)=p_{\bf w}(n+1)-p_{\bf w}(n)\ge 1$ for all $n$.
We only consider the case of upper frequencies, as the case of lower frequencies can be handled analogously. To prove this, we work with the restricted subword complexity function $p_{\bf w}^0(n)$, which counts the number of distinct length-$n$ subwords of ${\bf w}$ that occur infinitely often. 
We observe that if Inequality (\ref{eq:in}) holds for ${\bf w}$ with some positive constant $B$, then we in fact have
\begin{equation}
\label{eq:in0}
\Delta p^0_{\bf w}(n) := p^0_{\bf w}(n+1)-p^0_{\bf w}(n) \le p_{\bf w}(n+1)-p_{\bf w}(n)\end{equation} for all $n$. 
To see this, let $S_n$ denote the set of length-$n$ subwords of ${\bf w}$ and let $T_n$ denote the set of words in $S_n$ that only occur finitely many times in ${\bf w}$. Then each word in $T_n$ has at least one right prolongation in $S_{n+1}$, and since a word can only occur finitely often if each of its prolongations do, we have $|T_{n+1}|\ge |T_n|$. Then $p_{\bf w}^0(n)=|S_n|-|T_n|$, so 
$$\Delta p_{\bf w}^0(n) = |S_{n+1}| -|S_n| + |T_n|-|T_{n+1}| \le p_{\bf w}(n+1)-p_{\bf w}(n)=\Delta p_{\bf w}(n).$$

To prove our result, we note that if $u$ is not a subword of ${\bf w}$ or if it only occurs finitely many times then it has frequency $0$, and so it suffices to consider upper and lower frequencies of subwords of ${\bf w}$ that occur infinitely often. We fix a natural number $N$ and let $u_1,\ldots ,u_{p_{\bf w}^0(N)}$ denote the length-$N$ subwords of ${\bf w}$ that occur infinitely often in ${\bf w}$. 

For any finite word $u$, let 
\begin{equation}
R(u) = \{a \in \Sigma \colon ua \text{ occurs infinitely often in } {\bf w}\}
\end{equation}
 and 
 \begin{equation}
 L(u) = \{b \in \Sigma \colon bu \text{ occurs infinitely often in } {\bf w}\}
 \end{equation} denote the sets of infinitely occurring right and left extensions of $u$, respectively. We say $u$ is \emph{right-special} if $|R(u)| \ge 2$, and \emph{left-special} if $|L(u)| \ge 2$. We let $\mathrm{RS}(N)$ and $\mathrm{LS}(N)$ denote respectively the sets of right- and left-special words of length $N$ that occur infinitely often.

We now construct special subgraphs of the \emph{Rauzy graphs} (see \cite{Rauzy}) associated to ${\bf w}$. We let $\Gamma_N({\bf w})$ be the graph whose vertex set is precisely the set of vertices $V_N = \{u_1,\ldots ,u_{p_{\bf w}^0(N)}\}$. We draw a directed edge from $u_i$ to $u_j$ if there exist letters $a,b\in\Sigma$ such that $u_i a = b u_j$ and $u_i a$ occurs infinitely often in ${\bf w}$. We write $u_i\mapsto u_j$ for this directed edge, which corresponds to an out-arrow from $u_i$ and an in-arrow to $u_j$. 

Notice that the total number of edges in $\Gamma_N({\bf w})$ is exactly $p^0_{\bf w}(N+1) = \sum_{w \in V_N} |R(w)|$. Because $\sum_{w \in V_N} (|R(w)| - 1) = p^0_{\bf w}(N+1) - p^0_{\bf w}(N) = \Delta p^0_{\bf w}(N)$, the number of right-special vertices is bounded as follows:
\[ \sum_{w \in \mathrm{RS}(N)} (|R(w)| - 1) = \Delta p^0_{\bf w}(N). \]
In particular, at most $ \Delta p^0_{\bf w}(N)$ vertices have an out-degree greater than $1$. If $u_i$ is not right-special, it has a unique right extension $a_i \in \Sigma$ and a unique out-arrow $u_i \mapsto u_{\alpha(i)}$. Thus, for all but finitely many occurrences of $u_i$ in ${\bf w}$, it is immediately followed by $a_i$, yielding an occurrence of $u_{\alpha(i)}$. Hence, $\overline{{\rm freq}}_{u_i}({\bf w}) \le \overline{{\rm freq}}_{u_{\alpha(i)}}({\bf w})$ for all $u_i \notin \mathrm{RS}(N)$. By a symmetric argument on in-arrows, the number of left-special vertices satisfies $\sum_{v \in \mathrm{LS}(N)} (|L(v)| - 1) = \Delta p^0_{\bf w}(N)$, so $|\mathrm{LS}(N)| \le \Delta p^0_{\bf w}(N)$.

Following  
Balková--Pelantová \cite{Balkova}, we form the \emph{reduced Rauzy graph} $\tilde{\Gamma}_N({\bf w})$. Its vertex set $\tilde{V}_N$ is the collection of all vertices in $V_N$ that are either left-special or right-special (or both). The edges of $\tilde{\Gamma}_N({\bf w})$ correspond to \emph{simple paths} in $\Gamma_N({\bf w})$, which are directed paths $u_{i_0}\to u_{i_1}\to\cdots\to u_{i_m}$ of length $m \ge 1$ such that $u_{i_0}, u_{i_m} \in \tilde{V}_N$ and all internal vertices $u_{i_1},\ldots,u_{i_{m-1}}$ belong to $V_N \setminus \tilde{V}_N$.

Every internal vertex $u_{i_k}$ of a simple path has out-degree $1$ and in-degree $1$ in $\Gamma_N({\bf w})$. Thus, $\overline{{\rm freq}}_{u_{i_k}}({\bf w}) \le \overline{{\rm freq}}_{u_{i_{k+1}}}({\bf w})$ and $\overline{{\rm freq}}_{u_{i_{k+1}}}({\bf w}) \le \overline{{\rm freq}}_{u_{i_k}}({\bf w})$, meaning that upper frequencies are invariant along the interior of any simple path. We now claim that the number of distinct upper frequencies among all infinitely-occurring length-$(N+1)$ subwords is bounded above by the total number of edges in $\tilde{\Gamma}_N({\bf w})$.  We first observe that, since ${\bf w}$ is not eventually periodic, for every $j$ there must be at least one infinitely occurring subword of length $j$ that has at least two right continuations occurring as subwords infinitely often. Hence, starting from any vertex of $\Gamma_N({\bf w})$ and following out-arrows (respectively, in-arrows) one reaches a right-special (respectively, left-special) vertex after finitely many steps. In particular $\tilde{V}_N$ is not empty and each maximal simple path has both of its endpoints in $\tilde{V}_N$, so the map $\Phi$ below is defined on every infinitely-occurring length-$(N+1)$ factor.

To see this, we construct a map $\Phi$ from length-$(N+1)$ subwords of ${\bf w}$ that occur infinitely often to edges $\tilde{\Gamma}_N({\bf w})$ such that the preimage of each edge consists of words with the same upper frequency.  To construct our map, let $u$ be a word of length $N+1$ that occurs infinitely often. We let $v_0$ and $v_1$ denote respectively the length-$N$ prefix and length-$N$ suffix of $u$, then if $v_0$ and $v_1$ are both either left- or right-special then $v_0\to v_1$ is an edge in $\tilde{\Gamma}_N({\bf w})$ and we take $\Phi(u)$ to be this edge.  Since $v_0$ and $v_1$ are each left- or right-special, $u$ is the unique word of length $N+1$ that is mapped to this edge.
Alternatively, if $v_0$ or $v_1$ is neither left- nor right-special, then the path $v_0\to v_1$ extends uniquely to a path $$u_{i_0}\to u_{i_1}\to \cdots \to u_{i_k}\to v_0\to v_1\to u_{j_0} \to \cdots \to u_{j_m}$$ in which at most one of the subpaths
$u_{i_0}\to u_{i_1}\to \cdots \to u_{i_k}$ and $u_{j_0}\to \cdots \to u_{j_m}$ can be empty and such that only the initial and terminal vertices are either left- or right-special.
We then define $\Phi(u)$ to be the edge corresponding to this simple path.
Then one of $v_0$ and $v_1$ must be interior, and since all interior vertices have the same upper frequency, we see that all words that correspond to this path have the same upper frequency and so we have proven the claim.  

To finish the proof, we count the edges of $\tilde{\Gamma}_N({\bf w})$ in terms of the complexity difference $\Delta p^0_{\bf w}(N)$. Each right-special vertex $w$ is the source of exactly $|R(w)|$ simple paths, while a vertex $v$ that is left-special but not right-special is the source of exactly $1$ simple path. We thus obtain:
\begin{align*}
E(\tilde{\Gamma}_N) &= \sum_{w \in \mathrm{RS}(N)} |R(w)| + |\mathrm{LS}(N) \setminus \mathrm{RS}(N)| \\
&= \sum_{w \in \mathrm{RS}(N)} \bigl(|R(w)|-1\bigr) + |\mathrm{RS}(N)| + |\mathrm{LS}(N) \setminus \mathrm{RS}(N)| \\
&= \Delta p^0_{\bf w}(N) + |\mathrm{RS}(N) \cup \mathrm{LS}(N)|.
\end{align*}
Using the fact that $|\mathrm{RS}(N)| \le \Delta p^0_{\bf w}(N)$ and $|\mathrm{LS}(N)| \le \Delta p^0_{\bf w}(N)$, we obtain:
\[ E(\tilde{\Gamma}_N) \le \Delta p^0_{\bf w}(N) + |\mathrm{RS}(N)| + |\mathrm{LS}(N)| \le 3\Delta p^0_{\bf w}(N). \]
Together with the single value $0$ arising from factors that occur only finitely often, using Equation (\ref{eq:in0}) we have
there are at most $3\Delta p^0_{\bf w}(N)+1\le 3\Delta p_{\bf w}(N)+1$ distinct upper frequencies.

The result now follows, with the ``In particular'' clause following from Inequality 
\ref{eq:in}, which shows that we can take $C=3B+1$.
\end{proof}
\begin{remark} We observe that the above proof applies with minor modifications if we consider other notions of density such as logarithmic frequency or Banach frequency since intermediate vertices of simple paths of $\Gamma_N({\bf w})$ will again give subwords with identical logarithmic frequencies and Banach frequencies.
\end{remark}

\section{Examples}\label{sec:exam}
Borel's theorem \cite{Bor} shows that in some natural sense almost all words over a finite alphabet $\Sigma$ are normal and hence length-$d$ words all occur with the ``expected'' frequency of $|\Sigma|^{-d}$. On the other hand, it is reasonably straightforward to produce words in which the number of distinct frequencies of length-$d$ subwords is not uniformly bounded. For example, if ${\bf w}$ is an infinite normal word over the alphabet $\{0,1,2\}$, and $\phi: \{0,1,2\}^*\to \{a,b\}^*$ is the monoid homomorphism induced by $\phi(0)=a, \phi(1)=\phi(2)=b$, then since ${\bf w}$ is normal, each length-$d$ word over $\{0,1,2\}^*$ has frequency $1/3^d$ in ${\bf w}$. If we now consider ${\bf w}':=\phi({\bf w})\in \{a,b\}^{\omega}$ then the frequency of a length-$d$ word $u$ in $\{a,b\}^*$ in ${\bf w}'$ is easily seen to be 
$$3^{-d}\cdot \left| \left\{ v\in \{0,1,2\}^*\colon \phi(v)=u\right\}\right| = 3^{-d} \cdot 2^{u_b},$$ where $u_b$ denotes the number of occurrences of $b$ in the word $u$. In particular, there are $d+1$ distinct frequencies, namely $\{2^i/3^d \colon i=0,\ldots ,d\}$, of length-$d$ subwords of ${\bf w}'$. On the other hand, this example, while easy to produce, has a subword complexity function that grows exponentially in $n$ and produces a relatively small number of distinct frequencies of length-$n$ words as a function of $n$. We give somewhat more complicated examples, which show that the general behaviour can be very pathological.

\begin{example}
\label{exam:1}
Given a positive integer $k$, there is a right-infinite word ${\bf v}$ in which the number of distinct frequencies of length-$d$ subwords is at least $k^d$ for every $d\ge 1$.  
\end{example}
\begin{proof}
Let $\Sigma=\{0,1,\ldots ,k-1\}$. We order the non-trivial elements of $\Sigma^*$ degree-lexicographically, so $w\prec v$ if $|w|<|v|$ or if $|w|=|v|$ and $w$ is lexicographically smaller than $v$ when we read left-to-right and use the order $0\prec 1\prec \cdots \prec k-1$. Then we get an ordering 
$$0\prec 1\prec 2 \prec \cdots \prec k-1\prec 00 \prec 01 \cdots $$
and we let $w_i$ denote the $i$-th element on this list, where we begin with $w_0=0, w_1=1,\ldots$.  
We now let $4<p_1< p_2<\cdots $ be an increasing sequence of integers. We form the right-infinite word ${\bf v}=w_{q_1}xw_{q_2}xw_{q_3}x\cdots$ over $\Sigma\cup \{x\}$, in which for $i\ge 1$, $q_i = s$ where $s$ is the largest nonnegative integer for which $p_1\cdots p_s \mid i$. (By convention, we take $p_1\cdots p_s=1$ when $s=0$.) We let $c(w_i, w_j)$ denote the number of occurrences of $w_i$ in $w_j$. Then the number of occurrences of $w_i$ in the prefix $v_m:=w_{q_1}x w_{q_2} x \cdots x w_{q_m} x$ of ${\bf v}$ is equal to
$$\sum_{j=1}^m c(w_i, w_{q_j}),$$ since $w_i$ has no occurrences of $x$.
Since the number of $j$ for which $q_j=s$ is exactly $\lfloor m/p_1\cdots p_s\rfloor - \lfloor m/p_1\cdots p_{s+1}\rfloor$, we see that $v_m$ has 
$$\sum_{s=0}^{\infty} c(w_i,w_s) \cdot \left(\lfloor m/p_1\cdots p_s\rfloor - \lfloor m/p_1\cdots p_{s+1}\rfloor \right)$$
occurrences of $w_i$. Thus as $m$ tends to infinity, one can show that the number of occurrences of $w_i$ in $v_m$ is asymptotic to 
$$ m\cdot \sum_{s=0}^{\infty} c(w_i,w_s)\cdot \left(1/(p_1\cdots p_s)-1/(p_1\cdots p_{s+1})\right).$$
On the other hand, a similar analysis shows that the length of $v_m$ is asymptotic to 
$$m+\sum_{j=1}^m |w_{q_j}| \sim m\left(1 + \sum_{s=0}^{\infty} |w_s|\cdot 1/p_1\cdots p_s\right).$$
Since the ratio of the length of $v_m$ and the length of $v_{m+1}$ tends to one as $m\to\infty$, we see that the frequency of $w_i$ exists and is equal to 
$$\alpha_i:=\frac{\sum_{s=0}^{\infty} c(w_i,w_s)\cdot \left(1/(p_1\cdots p_s)-1/(p_1\cdots p_{s+1})\right)}{\left(1 + \sum_{s=0}^{\infty} |w_s|\cdot 1/p_1\cdots p_s\right)}.$$
Since the denominator in the above expression for each $\alpha_i$ is the same, to show that $\alpha_i\neq \alpha_j$ it suffices to show that the values
$\beta_i:=\sum_{s=0}^{\infty} c(w_i,w_s)\cdot \left(1/(p_1\cdots p_s)-1/(p_1\cdots p_{s+1})\right)$ are pairwise distinct for $i=0,1,\ldots$. But since $c(w_i,w_s)=1$ if $s=i$ and is $0$ for $s<i$, and since $c(w_i,w_s) \le |w_s| = O(\log_k(s+1))$, we see that if $p_{s+1}\ge \prod_{i\le s} p_i$ for all $s$, then we have the inequalities 
$1/(2p_1\cdots p_i) <\beta_i < 2/(p_1\cdots p_i)$. Thus the sequence $\beta_i$ is strictly decreasing with $i$, and so in particular, we get at least $k^d$ distinct subword frequencies of length-$d$ subwords of ${\bf v}$.  
\end{proof}

We now give a second example, which shows that once one moves outside of the realm of linearly bounded subword complexity, the conclusion to the ``In particular'' clause of Theorem \ref{thm:main} no longer holds in general.

\begin{example} \label{exam:2} Let $f(n)$ be a weakly increasing function that tends to infinity as $n\to\infty$. Then there is a right-infinite word ${\bf w}$ with $p_{\bf w}(n)=O(n f(n))$ such that the number of distinct upper frequencies of length-$d$ subwords is not uniformly bounded.
\end{example}

This construction is reasonably involved. We begin by forming a binary word
$${\bf v}=01^{
u_2(1)}01^{
u_2(2)}0 1^{
u_2(3)}\cdots,$$
where $
u_2$ is the $2$-adic valuation. Given a positive integer $n$, we let $v_n=01^{
u_2(1)}01^{
u_2(2)}\cdots 0 1^{
u_2(n)}$, whose length is given by:
\begin{equation}
|v_n| = n + u_2(1)+\cdots + u_2(n) \le 2n.
\end{equation}

Given a binary word $u$ we let $u'$ denote the word $\sigma(u)$, where $\sigma$ is the coding induced by $0\mapsto 1,\, 1\mapsto 0$. For a finite word $v$ and a subword $u$ of $v$, we write ${\rm occ}_u(v)$ for the number of occurrences of $u$ in $v$ and ${\rm freq}_u(v)={\rm occ}_u(v)/|v|$ for its frequency.

For a fixed $r$, for $n$ sufficiently large (with bound depending on $r$) and each $i\in\{2,\ldots,r\}$ there is a unique length-$n$ subword of ${\bf v}$ that begins $01^{n-i}0$, since the prefix $01^{n-i}0$ must occur in a subword of the form $01^{u_2(m)}0 1^{u_2(m+1)} 0 \cdots $ with $2^{n-i}\mid m$ and so the values of $u_2(m+1), u_2(m+2),\ldots ,u_2(m+r)$ are completely determined if $r<2^{n-r}\le 2^{n-i}$. Moreover, this unique length-$n$ word with prefix $01^{n-i}0 $occurs in ${\bf v}$ with frequency
$$2^{\,i-n}/4 \;=\; 2^{\,i-n-2}.$$ To see this, observe that the occurrences are in bijection with occurrences of $01^{u_2(m)}$ with $m$ being a multiple of $2^{n-i}$ but not of $2^{n-i+1}$.  As $N\to \infty$, the number of occurrences of such a block in the prefix $v_N$ is thus asymptotic to $N/2^{n-i} - N/2^{n-i+1}=N/2^{n-i+1}$ and since $|v_N|\sim N + u_2(1)+\cdots +u_2(N) \sim 2N$, we can see that within a length $M$ prefix, the number of occurrences is asymptotic to $M/2^{n-i+2}$, giving the claim.

(The restriction $i\ge 2$ is needed so that the prefix $01^{n-i}0$, which has length $n-i+2$, fits within a word of length $n$; uniqueness holds because, once $n-i$ is large compared to $i$, the $i$ letters following the maximal run $1^{n-i}$ are forced.) Since each $v_n$ is a prefix of ${\bf v}$ of length tending to infinity, the limsup of the upper frequency in $v_n$ of any fixed subword of ${\bf v}$ converges, as $n\to\infty$, to its upper frequency in ${\bf v}$.

We fix a weakly increasing sequence of positive integers $2\le r_0\le r_1\le r_2\le \cdots$, each a power of $2$, with $r_n\to\infty$ and
\begin{equation}
\label{eq:rn}
r_n \le f(n)\qquad {\rm for}~n~{\rm sufficiently~large.}
\end{equation}
Such a sequence exists because $f$ is weakly increasing and tends to infinity: take $r_n$ to be the largest power of two not exceeding $\max(2,f(n))$.

Given a finite binary word $v$ and two finite words $u_0,u_1$ over a possibly different alphabet, we define $v(u_0,u_1)$ to be the word obtained by substituting $0\mapsto u_0$ and $1\mapsto u_1$ in $v$. We now make a right-infinite word ${\bf w}\in \{x,y\}^{\omega}$ as follows. We let $w_1=v_{r_0}(x,y)$ and, having defined $w_n$, we set $w_{n+1}= v_{r_n}(w_n, w_n')$, where $w_n'$ is obtained from $w_n$ by applying the coding $x\mapsto y,\, y\mapsto x$. We take $v_i'$ to be the word obtained by applying $\sigma$ to $v_i$.

Since each $v_i$ is a prefix of ${\bf v}$, which begins with $0$, we see that $w_n$ is a prefix of $w_{n+1}$ for each $n$, and so we can take the limit of the $w_n$ to obtain a right-infinite word ${\bf w}$. By construction ${\bf w}\in\{w_n,w_n'\}^{\omega}$ for every $n\ge 1$.

\begin{lemma}
For all $n$ the number of occurrences of $w_n$ in $w_nw_n$, $w_nw_n'$, $w_n'w_n$, and $w_n'w_n'$ is respectively $2,1,1,$ and $0$. In particular, whenever a non-trivial word $u\in \{w_n,w_n'\}^*$ occurs as a subword of an infinite word $y_0y_1y_2\cdots $ in which $y_i\in \{w_n,w_n'\}$ for each $i$, there exist $p,q$ with $p\le q$ such that $u=y_p\cdots y_q$.
\label{lem:f0}
\end{lemma}
\begin{proof}
We only prove the claim for occurrences in $w_nw_n$ with the other cases being slightly simpler. We first claim that for $p\ge 1$, $v_{2^p}$ occurs exactly twice in $v_{2^p} v_{2^p}$. To see this, if this were not the case, there would be some $i\in \{1,2,\ldots ,2^p-1\}$ such that 
$v_{i,2^p}:=01^{u_2(i)}0\cdots 01^{u_2(2^p)}01^{u_2(1)}0\cdots 01^{u_2(i-1)} = v_{2^p}$. But by construction $1^{p}$ occurs exactly once in each of these words and if $i>0$ and $i<2^p$ then these occurrences cannot begin in the same positions of $v_{i,2^p}$ and $v_{2^p}$.  
We now show by induction that $w_n$ cannot occur more than twice in $w_nw_n$ with the base case following from the claim. Now assume that the claim holds for positive integers less than or equal to $m$.
Then $w_{m+1}=v_{r_{m}}(w_m,w_m')$. If $w_{m+1}$ occurs somewhere in $w_{m+1}w_{m+1}$ other than in the two obvious occurrences then our induction hypothesis gives that $w_m$ cannot occur in any of the overlaps of $w_m^2, w_m w_m', w_m'w_m, w_m'w_m'$ and so this yields an occurrence of $v_{r_{m}}$ in $v_{r_{m}} v_{r_{m}}$ not occurring at the prefix or the suffix, which we've shown cannot occur. And so the first claim follows by induction. By symmetry an analogous result holds for occurrences of $w_n'$ in $w_nw_n$, $w_nw_n'$, $w_n'w_n$, and $w_n'w_n'$. Now suppose that we have a right-infinite word $y_0y_1y_2\cdots $ in which $y_i\in \{w_n,w_n'\}$ for each $i$, and suppose that $u\in \{w_n,w_n'\}^*$ is a subword (where we are considering it as a binary word). Then if $u$ occurs with some starting position that is not equal to the starting position of some $y_p$ then there is some occurrence of $w_n$ or $w_n'$ that occurs as a subword of one of $w_nw_n$, $w_nw_n'$, $w_n'w_n$, $w_n'w_n'$ without being either a prefix or suffix. But we have shown that this is impossible.
\end{proof}

For each $n$ the word ${\bf w}$ has a unique factorization into blocks from $\{w_n,w_n'\}$; identifying $w_n$ with $0$ and $w_n'$ with $1$ turns this factorization into a binary sequence ${\bf z}^{(n)}\in\{0,1\}^{\omega}$. Since $w_{n+1}=v_{r_n}(w_n,w_n')$ and ${\bf w}\in\{w_{n+1},w_{n+1}'\}^{\omega}$, we have ${\bf z}^{(n)}\in\{v_{r_n},v_{r_n}'\}^{\omega}$. We let
\begin{equation}
\label{eq:alpha}
\alpha_n=\limsup_{N\to\infty}\frac{\#\{v_{r_n}\text{-blocks among the first }N~{\rm blocks~ of~} {\bf z}^{(n)}\}}{N},
\end{equation}
\begin{equation}
\alpha_n'=\limsup_{N\to\infty}\frac{\#\{v_{r_n}'\text{-blocks among the first }N~{\rm blocks~of~}{\bf z}^{(n)}\}}{N}.
\end{equation}
Since every block in our decomposition of ${\bf z}^{(n)}$ is either $v_{r_n}$ or $v_{r_n}'$, we have $\alpha_n+\alpha_n'\ge 1$.

\begin{lemma}\label{lem:freq2}
Let $n$ be a positive integer, and let $u$ be a factor of $v_{r_n}$ that contains the factor $111$. Then the upper frequency of $u$ in the sequence ${\bf z}^{(n)}$ (occurrences per symbol) lies in the closed interval
$$\bigl[\,{\rm freq}_u(v_{r_{n}})\cdot \alpha_n,\; {\rm freq}_u(v_{r_{n}})\cdot \alpha_n+ |u|/|v_{r_{n}}|\,\bigr].$$
\end{lemma}
\begin{proof}
Consider a prefix $y$ of ${\bf z}^{(n)}$ consisting of $B$ blocks from $v_{r_n}, v_{r_n}'$, and let $A$ denote the number of blocks that are $v_{r_n}$; it has $B\,|v_{r_n}|$ symbols. Each occurrence of $u$ in ${\bf z}^{(n)}$ either lies inside a single block or lies in an overlap between two consecutive blocks. Since $v_{r_n}$ has no run of three $0$'s, $v_{r_n}'$ has no run of three $1$'s and so since $u$ contains $111$, no occurrence of $u$ lies inside a $v_{r_n}'$ block. Hence the number of within-block occurrences of $u$ in $y$ is exactly $A\cdot {\rm occ}_u(v_{r_n})$, while the number of occurrences of $u$ in an overlap is at most $|u|$ per boundary. Hence we get at most $|u|\,B$ boundary occurrences in total. Dividing by $B\,|v_{r_n}|$ and using ${\rm freq}_u(v_{r_n})={\rm occ}_u(v_{r_n})/|v_{r_n}|$, the number of occurrences of $u$ per symbol in this prefix lies in
$$\bigl[\,(A/B)\cdot{\rm freq}_u(v_{r_n}),\; (A/B)\cdot{\rm freq}_u(v_{r_n})+|u|/|v_{r_n}|\,\bigr].$$
Taking the limsup as $B\to\infty$ gives the result.
\end{proof}

\begin{remark}
An analogous statement holds when $u$ is a factor of $v_{r_n}'$ containing $000$ (which $v_{r_n}$ cannot contain), with $\alpha_n$ replaced by $\alpha_n'$.
\end{remark}

\begin{lemma} The number of distinct upper frequencies of length-$d$ subwords of ${\bf w}$ is not uniformly bounded in $d$.
\label{lem:f2}
\end{lemma}
\begin{proof}
By $\alpha_n+\alpha_n'\ge 1$, either $\alpha_n\ge 1/2$ for infinitely many $n$, or $\alpha_n'\ge 1/2$ for infinitely many $n$. We treat the first case; the second is handled analogously, using the Remark.

Fix an integer $r\ge 2$. Recall that for sufficiently large $m$ and for each $i\in\{2,\ldots,r\}$ there is a unique factor $u_{i,m}$ of ${\bf v}$ of length $m$ beginning $01^{m-i}0$, and ${\rm freq}_{u_{i,m}}(v_n)\to 2^{\,i-m}/4$ as $n\to\infty$. Fix an $m\ge 3+r$ having the uniqueness property described above. Since $u_{i,m}$ begins $01^{m-i}0$ with $m-i\ge m-r\ge 3$, it contains $111$. We now choose $n$ large enough so that the following hold:
\begin{itemize}
\item[(a)] $\alpha_n\ge 1/2$;
\item[(b)] $v_{r_n}$ contains each $u_{i,m}$ and ${\rm freq}_{u_{i,m}}(v_{r_n})\in\bigl(0.95\cdot 2^{\,i-m}/4,\; 1.05\cdot 2^{\,i-m}/4\bigr)$;
\item[(c)] $|u_{i,m}|/|v_{r_n}|=m/|v_{r_n}|<2^{-m}/100$.
\end{itemize}
Then by Lemma~\ref{lem:freq2}, the upper frequency of $u_{i,m}$ in ${\bf z}^{(n)}$ lies in
$$I_{i,m}:=\bigl[\,0.95\,\alpha_n\,2^{\,i-m}/4,\; 1.05\,\alpha_n\,2^{\,i-m}/4 + 2^{-m}/100\,\bigr].$$
We claim that the intervals $I_{i,m}$, $i=2,\ldots,r$, are pairwise disjoint. To see this, suppose $I_{i,m}\cap I_{j,m}$ is nonempty with $i<j$. As the intervals are ordered by their left endpoints, this forces
$$0.95\,\alpha_n\,2^{\,j-m}/4 \le 1.05\,\alpha_n\,2^{\,i-m}/4 + 2^{-m}/100.$$
Multiplying by $400\cdot 2^{m}$ gives $95\,\alpha_n 2^{j} \le 105\,\alpha_n 2^{i}+4$, i.e.
$$\alpha_n\bigl(95\cdot 2^{j}-105\cdot 2^{i}\bigr)\le 4.$$
Since $j\ge i+1$, we have $95\cdot 2^{j}-105\cdot 2^{i}\ge 2^{i}(190-105)=85\cdot 2^{i}\ge 85$, so $\alpha_n\le 4/85<1/2$, contradicting (a). Hence the $I_{i,m}$ are pairwise disjoint.

Finally, occurrences of $u_{i,m}$ in ${\bf z}^{(n)}$ are in one-to-one correspondence with occurrences of the factor $u_{i,m}(w_n,w_n')$ in ${\bf w}$, and we see that the upper frequency (per letter of ${\bf w}$) of $u_{i,m}(w_n,w_n')$ equals $1/|w_n|$ times the upper frequency of $u_{i,m}$ in ${\bf z}^{(n)}$, and these latter values lie in the pairwise disjoint intervals $I_{i,m}$. Since $1/|w_n|$ is a common positive factor, the $r-1$ words $u_{i,m}(w_n,w_n')$, all of length $m\,|w_n|$, have pairwise distinct upper frequencies. Since $r$ is arbitrary, the number of distinct upper frequencies of length-$d$ subwords of ${\bf w}$ is not bounded in $d$.
\end{proof}

\begin{lemma} Adopting the above notation, we have $p_{\bf w}(n)=O(n f(n))$.
\label{lem:f1}
\end{lemma}
\begin{proof}
Fix $n$ large enough that $n\ge|w_1|$, and let $i$ be the unique index with $|w_i|\le n<|w_{i+1}|$. Since ${\bf w}\in \{w_{i+1}, w_{i+1}'\}^{\omega}$ and $|w_{i+1}|,|w_{i+1}'|>n$, every length-$n$ subword of ${\bf w}$ is a subword of one of $w_{i+1}w_{i+1}$, $w_{i+1}w_{i+1}'$, $w_{i+1}'w_{i+1}$, $w_{i+1}'w_{i+1}'$. As there are at most $n$ starting positions where a length-$n$ subword can occur in an overlap, we see
$$p_{\bf w}(n)\le 4n + p_{w_{i+1}}(n)+p_{w_{i+1}'}(n)= 4n+2\,p_{w_{i+1}}(n),$$
where the last equality holds because $w_{i+1}'$ is obtained by applying a bijective coding to $w_{i+1}$.

Now $w_{i+1}=v_{r_i}(w_i,w_i')$. Let $c=\lceil n/|w_i|\rceil$. Viewing $w_{i+1}$ as a word over $w_i,w_i'$, we see that a length-$n$ subword of $w_{i+1}$ overlaps with at most $c+1$ blocks of $w_i$ and $w_i'$. Thus it is determined by a length-$(c+1)$ factor of $v_{r_i}$ together with a starting offset in $\{0,\ldots,|w_i|-1\}$ to indicate where in the first block our occurrence begins. The number of length-$(c+1)$ factors of $v_{r_i}$ is at most $|v_{r_i}|\le 2r_i$. Hence, using $|w_i|\le n$,
$$p_{w_{i+1}}(n)\le |w_i|\cdot 2r_i\le n\,2r_i.$$
Therefore $p_{\bf w}(n)\le 4n+2n(2r_i)=(4r_i+4)\,n$.

Finally, $|w_{j+1}|=|v_{r_j}||w_j|\ge 3|w_j|$, so $|w_i|\ge 3^{i}$; as $|w_i|\le n$, this gives $i\le \log_3 n\le n$, and since $(r_n)$ is weakly increasing, $r_i\le r_n$. Thus
$$p_{\bf w}(n)\le (4r_n+4)\,n.$$
By \eqref{eq:rn}, $r_n\le f(n)$ for $n$ sufficiently large, so $p_{\bf w}(n)\le (4f(n)+4)\,n=O(n f(n))$.
\end{proof}

\begin{proof}[Proof of Example \ref{exam:2}]
Lemmas \ref{lem:f1} and \ref{lem:f2} show that ${\bf w}$ has the desired properties.
\end{proof}

\section*{Acknowledgments}
We thank Jeffrey Shallit and Jean-Paul Allouche for many helpful comments and suggestions and in particular, we thank Jean-Paul Allouche for alerting us to the work of
Boshernitzan and the work of Balkov\'a and Pelantov\'a.

\end{document}